\documentclass[11pt]{amsart}
\usepackage{amsmath, amssymb}
\usepackage{epsfig}

\theoremstyle{plain}
\newtheorem{Theorem}{Theorem}

\newtheorem{Lemma}{Lemma}
\newtheorem{Proposition}{Proposition}

\theoremstyle{definition}
\newtheorem{Remark}{Remark}

\DeclareMathOperator{\ch}{ch} \DeclareMathOperator{\sh}{sh}

\title[Hirota-Kimura type discretization]
{Hirota-Kimura type discretization of the classical nonholonomic Suslov problem}
\author{Vladimir Dragovi\'c and Borislav Gaji\'c }

\email{vladad@mi.sanu.ac.yu and gajab@mi.sanu.ac.yu}

\begin{document}
\begin{abstract}
We constructed Hirota-Kimura type discretization of the classical
nonholonomic Suslov problem of motion of rigid body fixed at a
point. We found a first integral proving integrability. Also, we
have shown that discrete trajectories asymptotically tend to a
line of discrete analogies of so-called steady-state rotations.
The last property completely corresponds to well-known
property of the continuous Suslov case. The explicite formulae for
solutions are given. In $n$-dimensional case we give discrete
equations.
\end{abstract}
\maketitle \tableofcontents
\section{Introduction}

There are several methods for constructing discrete counterparts of integrable dynamical systems in classical mechanics.
One is well-known Veselov-Moser discretization (see \cite{MV}). By this method, based on discrete variational principle,
many of discrete integrable systems are found (see \cite{Sur}). In  cases with discrete Lagrangian,
corresponding  discrete map is Poisson with
respect to certain Poisson structure. Usually the discrete map is multi-valued. Recently, Hirota and Kimura
 constructed explicit integrable discretizations of the Euler and the Lagrange cases of motion of a heavy rigid
body fixed at a point using Hirota's bilinear method \cite{HK1,HK2}. They found first integrals of motion and
they solved equations in terms of elliptic functions. Suris and Petrera in \cite{PS} found bi-Hamiltonian structure for
discrete Euler top.

Our goal is to apply Hirota-Kimura method to the classical nonholonomic Suslov problem. Suslov in \cite{Sus} considered motion of
a rigid body fixed at a point with projection of angular velocity to an axis fixed in the body equal to zero. Solutions
in terms of trigonometric and exponential functions are given as well in \cite{Sus}. There are various generalizations of
classical Suslov problem: integrable potential perturbations (\cite{K,DGJ}) and higher-dimensional generalizations(\cite{FK, J, J1, ZB}).
In \cite{FZ} Fedorov and Zenkov presented certain discretization of Suslov problem based on the Veselov-Moser discretization and its
extension to nonholonomic cases suggested in \cite{CM, LMS}. Obtained map is multi-valued. Discrete trajectories asymptotically tend to
discrete analogies of so-called steady-state rotations, the property characteristic  for continuous case considered by Suslov.
Moreover, Fedorov and Zenkov \cite{FZ} gave the equations in
$n$-dimensional case and proved similar asymptotic behavior as in three-dimensional case.

In the present paper we use the Hirota-Kimura method
and get explicit discretization of the reduced Suslov problem on the linear subspace of algebra $so(3)$.
The corresponding map is not multi-valued. We gave one first integral which appears to be enough for integration.
Using a linear change of variables, we have found explicite solutions
in terms of exponentional functions. (Notice, that in
the case of discretization given in \cite{FZ}, the explicite solutions
are not known.)
Presented discrete version of Suslov case has similar asymptotic behavior as the continuous one.

The paper is organized as follows. In Section 2 basic facts about the classical Suslov problem are given. In Section 3 we present
discrete equations of Hirota-Kimura type and we construct first integral of motion. Integration procedure for discrete equations
is performed in Section 4. In Section 5 we give the discrete equations for $n$-dimensional Suslov case.

\section{A brief account of the classical Suslov problem}

The classical nonholonomic Suslov problem is defined
in \cite{Sus}.  Configuration space is Lie group
$SO(3)$. In a bases fixed in the body, the equations of the motion are:
\begin{equation}
\begin{aligned}
\dot{\bold M}={\bold M}\times{\bold \Omega}+\lambda {\bold a}\\
\langle {\bold a}, {\bold \Omega} \rangle=0.
\end{aligned}
\label{s}
\end{equation}
Here ${\bold \Omega}$ is the angular velocity, and ${\bold M}=I{\bold \Omega}$ is the angular momentum,
 $I$ is the inertia operator, ${\bold a}$ is a unit vector fixed in the body and $\lambda$ is the Lagrange multiplier.
In a bases chosen such that ${\bold a}=(0,0,1)$ and
$$
I=\left[ \begin{matrix} I_1&0&I_{13}\\
                0& I_2&I_{23}\\
                I_{13}&I_{23}&I_3
\end{matrix}
\right],
$$
the equations \eqref{s} become:
\begin{equation}
\begin{aligned}
I_1\dot \Omega_1&=-I_{13}\Omega_1\Omega_2-I_{23}\Omega_2^2\\
I_2\dot \Omega_2&=I_{13}\Omega_1^2+I_{23}\Omega_1\Omega_2\\
0&=-I_{13}\dot\Omega_1-I_{23}\dot\Omega_2+(I_1-I_2)\Omega_1\Omega_2+\lambda\\
\Omega_3&=0.
\end{aligned}
\label {s1}
\end{equation}

The first two equations are closed in $\Omega_1$ and $\Omega_2$. After solving them one finds
the Lagrange multiplier $\lambda$ as a function of time from the
third equation. Hence, for complete integrability by quadratures, only one first integral of motion is necessary.
This is the energy integral as follows from
\eqref{s1} easily. Suslov in \cite{Sus} gave solutions of the system in terms of trigonometric and
exponential functions.
He observed a remarkable fact (as the referee observed to be known before to Walker and Routh):
the motion of the body asymptotically tends to a line of rotations
with constant angular velocities which satisfy $I_{13}\Omega_1+I_{23}\Omega_2=0$.

\section{Hirota-Kimura type discretization of the Suslov problem}

In the spirit of Hirota-Kimura, a discrete counterpart of the first two equations and nonholonomic constraint of \eqref{s} is:
\begin{equation}
\begin{aligned}
I_1&(\widetilde{\Omega}_1-\Omega_1)+I_{13}(\widetilde{\Omega}_3-\Omega_3)=\epsilon\Big[\frac{I_2}{2}(\widetilde{\Omega}_2\Omega_3+\Omega_2\widetilde{\Omega}_3)\\
&+I_{23}\Omega_3\widetilde{\Omega}_3-
\frac{I_3}{2}(\widetilde{\Omega}_2\Omega_3+\Omega_2\widetilde{\Omega}_3)-\frac{I_{13}}{2}(\widetilde{\Omega}_1\Omega_2+\Omega_1\widetilde{\Omega}_2)-
I_{23}\Omega_2\widetilde{\Omega}_2\Big]\\
I_2&(\widetilde{\Omega}_2-\Omega_2)+I_{23}(\widetilde{\Omega}_3-\Omega_3)=\epsilon\Big[\frac{I_3}{2}(\widetilde{\Omega}_1\Omega_3+\Omega_1\widetilde{\Omega}_3)\\
&-I_{13}\Omega_3\widetilde{\Omega}_3-
\frac{I_1}{2}(\widetilde{\Omega}_1\Omega_3+\Omega_1\widetilde{\Omega}_3)+\frac{I_{23}}{2}(\widetilde{\Omega}_1\Omega_2+\Omega_1\widetilde{\Omega}_2)+
I_{13}\Omega_1\widetilde{\Omega}_1\Big]\\
\widetilde{\Omega}_3&=-\Omega_3.
\end{aligned}
\label{ds}
\end{equation}
Here $\Omega_i=\Omega_i(t)$, $\widetilde\Omega_i=\Omega_i(t+\epsilon)$ and $\epsilon$ is the time step.
The limit when $\epsilon$ goes to $0$ should reconstruct equations \eqref{ds} of the continuous Suslov problem \eqref{s1}.
Thus one concludes that $\Omega_3$ should be equal to zero, and equations \eqref{ds} become:
\begin{equation}
\begin{aligned}
I_1(\widetilde{\Omega}_1-\Omega_1)&=\epsilon\Big[-\frac{I_{13}}{2}(\widetilde{\Omega}_1\Omega_2+\Omega_1\widetilde{\Omega}_2)-
I_{23}\Omega_2\widetilde{\Omega}_2\Big]\\
I_2(\widetilde{\Omega}_2-\Omega_2)&=\epsilon\Big[\frac{I_{23}}{2}(\widetilde{\Omega}_1\Omega_2+\Omega_1\widetilde{\Omega}_2)+
I_{13}\Omega_1\widetilde{\Omega}_1\Big]\\
\Omega_3&=0.
\end{aligned}
\label{ds1}
\end{equation}
Since these equations are linear in $\widetilde{\Omega}_i$, the map defined by \eqref{ds1} is
explicit and unique-valued:
$$
\left[\begin{matrix} \widetilde{\Omega}_1\\
                    \widetilde{\Omega}_2
\end{matrix}\right]=
\left[\begin{matrix}1+\frac{\epsilon I_{13}}{2I_1}\Omega_2& \frac{\epsilon I_{13}}{2I_1}\Omega_1+\frac{\epsilon I_{23}}{I_1}\Omega_2\\
-\frac{\epsilon I_{13}}{I_2}\Omega_1-\frac{\epsilon I_{23}}{2I_2}\Omega_2&1-\frac{\epsilon I_{23}}{2I_2}\Omega_1
\end{matrix}\right]^{-1}
\left[\begin{matrix}\Omega_1\\ \Omega_2
\end{matrix}\right],
$$
giving
\begin{equation}
\begin{aligned}
\widetilde{\Omega}_1&=\frac{1}{\Delta}(\Omega_1-\frac{\epsilon I_{23}}{2I_2}\Omega_1^2-
\frac{\epsilon I_{13}}{2I_1}\Omega_1\Omega_2-\frac{\epsilon I_{23}}{I_1}\Omega_2^2)\\
\widetilde{\Omega}_2&=\frac{1}{\Delta}(\Omega_2+\frac{\epsilon I_{13}}{2I_1}\Omega_2^2+
\frac{\epsilon I_{23}}{2I_2}\Omega_1\Omega_2+\frac{\epsilon I_{13}}{I_2}\Omega_1^2),
\end{aligned}
\label{ds2}
\end{equation}
where
$$
\Delta=\left(1+\frac{\epsilon I_{13}}{2I_1}\Omega_2\right)\left(1-\frac{\epsilon I_{23}}{2I_2}\Omega_1\right)+
\frac{\epsilon^2}{I_1I_2}\left(\frac{I_{13}\Omega_1}{2}+I_{23}\Omega_2\right)\left(I_{13}\Omega_1+\frac{I_{23}\Omega_2}{2}\right).
$$
As we have already mentioned, in continuous case equations \eqref{s1} have the
energy integral. But for considered discretization the energy is not an integral anymore, but there exists a first integral
as it can be seen from the following statement.
\begin{Lemma}{\label{l1}} The function
\begin{equation}
F=\frac{I_1\Omega_1^2+I_2\Omega_2^2}{4I_1I_2+\epsilon^2(I_{13}\Omega_1+I_{23}\Omega_2)^2}
\label{i1}
\end{equation}
is a first integral of equations \eqref{ds2}
\end{Lemma}
\begin{proof} Proof follows by direct calculations.
\end{proof}
In the limit when $\epsilon$ goes to zero, integral \eqref{i1} tends to the energy integral divided by constant.

\section{Integration}

In order to integrate discrete Suslov equations, we introduce new coordinates:
\begin{equation}
x=I_{13}\Omega_1+I_{23}\Omega_2,\ \ \ y=I_{23}I_{1}\Omega_1-I_{13}I_2\Omega_2.
\label{xy}
\end{equation}
The Jacobian of the change of coordinates \eqref{xy} is:
$$
-(I_{13}^2I_2+I_{23}^2I_1).
$$
Thus, it is equal to zero only in the case $I_{13}=I_{23}=0$ when
$\widetilde{\Omega}_1=\Omega_1$ and $\widetilde{\Omega}_2=\Omega_2$, giving equilibrium position.

In the new coordinates \eqref{xy} equations \eqref{ds1} are:
\begin{equation}
\begin{aligned}
\widetilde x-x&= \frac{\epsilon}{2I_1I_2}(\widetilde xy+x\widetilde y)\\
\widetilde y-y&= -\epsilon \widetilde xx.
\end{aligned}
\label{xy1}
\end{equation}
The first integral \eqref{i1} becomes
\begin{equation}
F=\frac{I_1I_2x^2+y^2}{4I_1I_2+\epsilon^2x^2}.
\label{i2}
\end{equation}
The curve $F(x,y)=h$ can be parameterized by introducing:
\begin{equation}
\begin{aligned}
x&=2\sqrt{\frac{I_1I_2h}{I_1I_2-h\epsilon\cos^2\phi}}\cos\phi,\\
y&=2I_1I_2\sqrt{\frac{h}{I_1I_2-h\epsilon\cos^2\phi}}\sin\phi.
\end{aligned}
\label{phi}
\end{equation}
Putting \eqref{phi} into the second equation of \eqref{xy1} and denoting
$$
u=\sqrt{\frac{I_1I_2}{I_1I_2-h\epsilon^2}}\tan\phi
$$
we get:
\begin{equation}
u\sqrt{\widetilde{u}^2+1}-\widetilde{u}\sqrt{u^2+1}=\frac{2\epsilon\sqrt{I_1I_2h}}{I_1I_2-h\epsilon^2}.
\label{u}
\end{equation}
Let us suppose the form of a solution of \eqref{u}: $u(n)=\sh(k_1(t_0+n\epsilon)+k_2)$, where $k_1$ and $k_2$ are constants. By plugging
the form into $\eqref{u}$, one gets:
\begin{equation}
\sh(-k_1\epsilon)=\frac{2\epsilon\sqrt{I_1I_2h}}{I_1I_2-h\epsilon^2}.
\label{k1}
\end{equation}
So, we have
\begin{Proposition}\label{t1} Let constant $k_1$ satisfy \eqref{k1} and $k_2$ be arbitrary
constant. Then the function $u(n)=\sh(k_1(t_0+n\epsilon)+k_2)$ gives solutions of equation \eqref{u}.
\end{Proposition}

As a consequence we have the following statement:
\begin{Theorem} Let constant $k_1$ satisfy \eqref{k1} and $k_2$ be arbitrary
constant. The functions:
$$
\begin{aligned}
\Omega_{1}(n)&= \frac{2\sqrt{h}}{(I_{13}^2I_1+I_{13}^2I_{2})\ch(k_1(t_0+n\epsilon)+k_2)}\left(\frac{I_{13}I_2}{\sqrt{I_1I_2-h\epsilon^2}}+I_{23}\sh(k_1(t_0+n\epsilon)+k_2)\right)\\
\Omega_{2}(n)&= \frac{2\sqrt{h}}{(I_{13}^2I_1+I_{13}^2I_{2})\ch(k_1(t_0+n\epsilon)+k_2)}\left(\frac{I_{23}I_1}{\sqrt{I_1I_2-h\epsilon^2}}-I_{13}\sh(k_1(t_0+n\epsilon)+k_2)\right)\\
\end{aligned}
$$
give the solutions of equations \eqref{ds2}.
\end{Theorem}
\begin{proof} From Proposition {\ref{t1}} and \eqref{phi} one gets:
$$
x(n)=2\sqrt{\frac{hI_1I_2}{I_1I_2-h\epsilon^2}}\frac{1}{\ch(k_1(t_0+n\epsilon)+k_2)},\,\, y(n)=2\sqrt{hI_1I_2}\frac{\sh(k_1(t_0+n\epsilon)+k_2)}{\ch(k_1(t_0+n\epsilon)+k_2)}.
$$
Proof follows from \eqref{xy}.
\end{proof}
\begin{Proposition}
 The discrete trajectories of motion of the body asymptotically tend to a line of discrete analogies of so-called steady-state
 rotations that satisfy:
$$
I_{13}\Omega_1+I_{23}\Omega_2=0.
$$
\end{Proposition}
\begin{proof}
In the limit $n\longrightarrow \pm\infty$, we have that $x$ goes to zero giving that
$I_{13}\Omega_1+I_{23}\Omega_2$ goes to zero.
\end{proof}

The last statement is illustrated by the following pictures. In all four cases, the line $I_{13}\Omega_1+I_{23}\Omega_2=0$ is
clearly indicated. Following parameters are chosen for picture 1: $\epsilon=0.2, I_1=4, I_2=1, I_{13}=-0.5, I_{23}=-0.3$,
for picture 2: $\epsilon=0.2, I_1=4, I_2=3, I_{13}=-0.4, I_{23}=-0.2$, for picture 3: $\epsilon=0.02, I_1=4, I_2=2, I_{13}=0, I_{23}=-0.2$
and for picture 4: $\epsilon=1, I_1=3, I_2=3, I_{13}=-0.2, I_{23}=-0.2$.

{\epsfig{figure=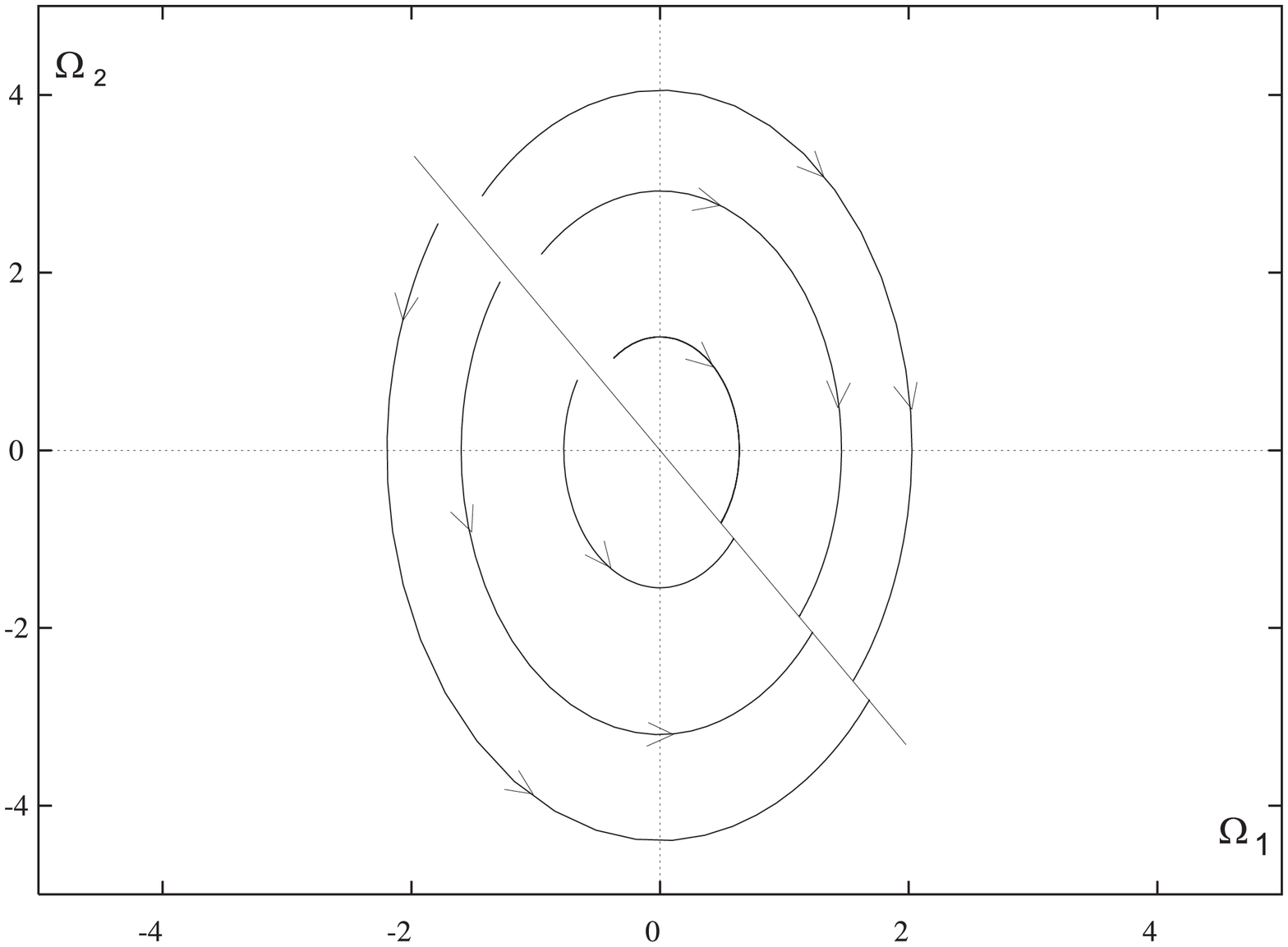, height=6.5cm}}

{\epsfig{figure=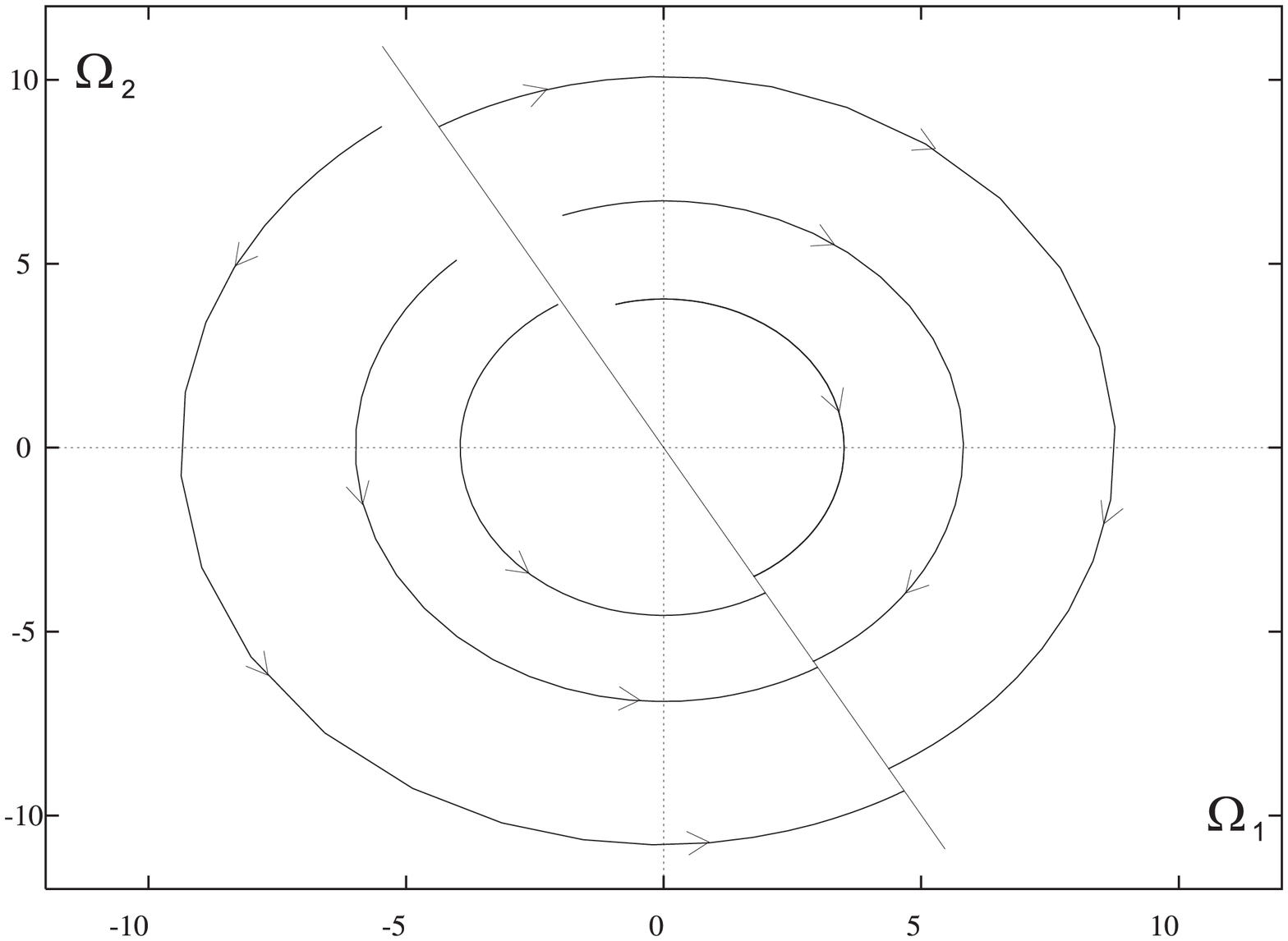, height=6.5cm}}

{\epsfig{figure=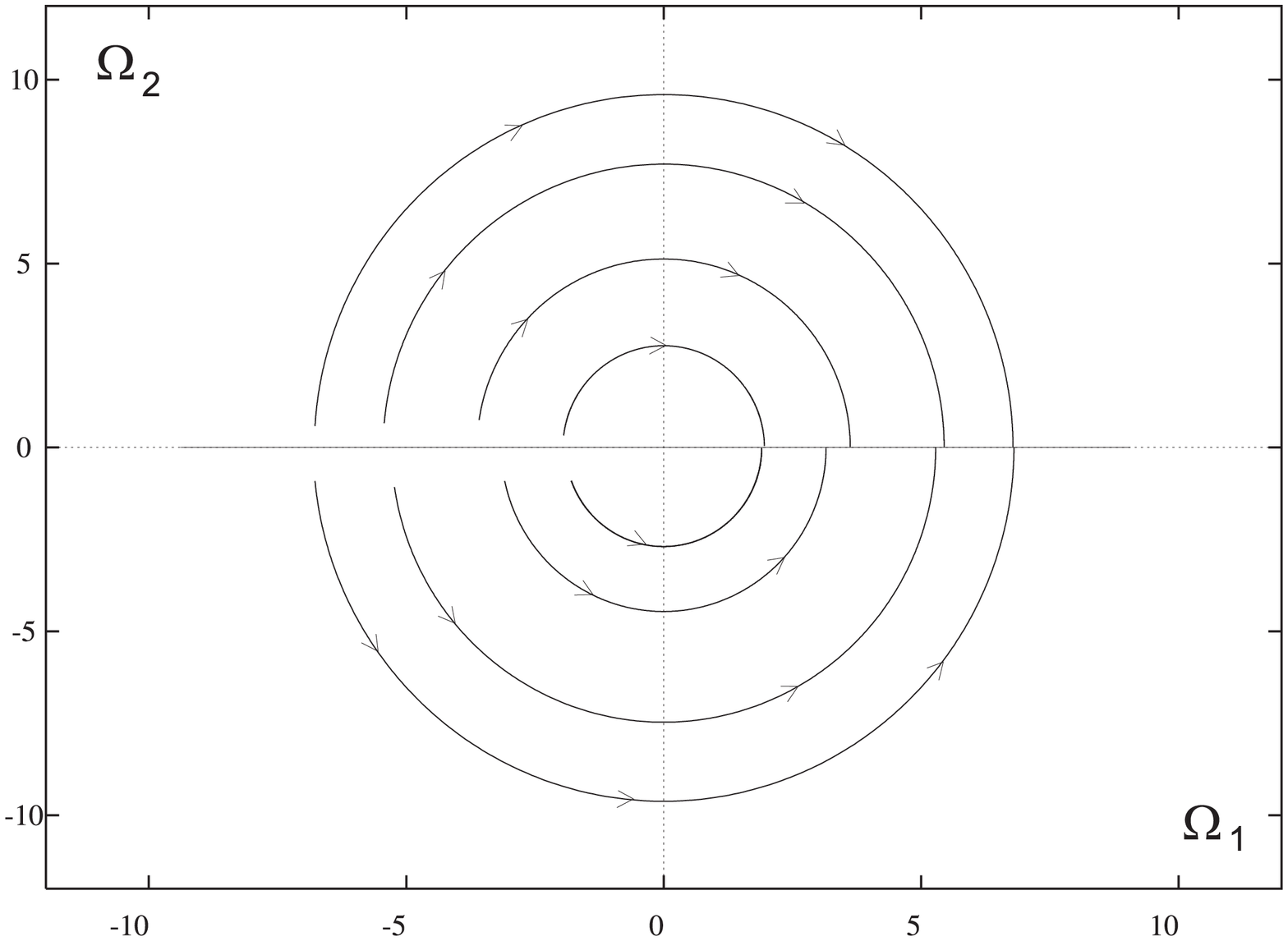, height=6.5cm}}

{\epsfig{figure=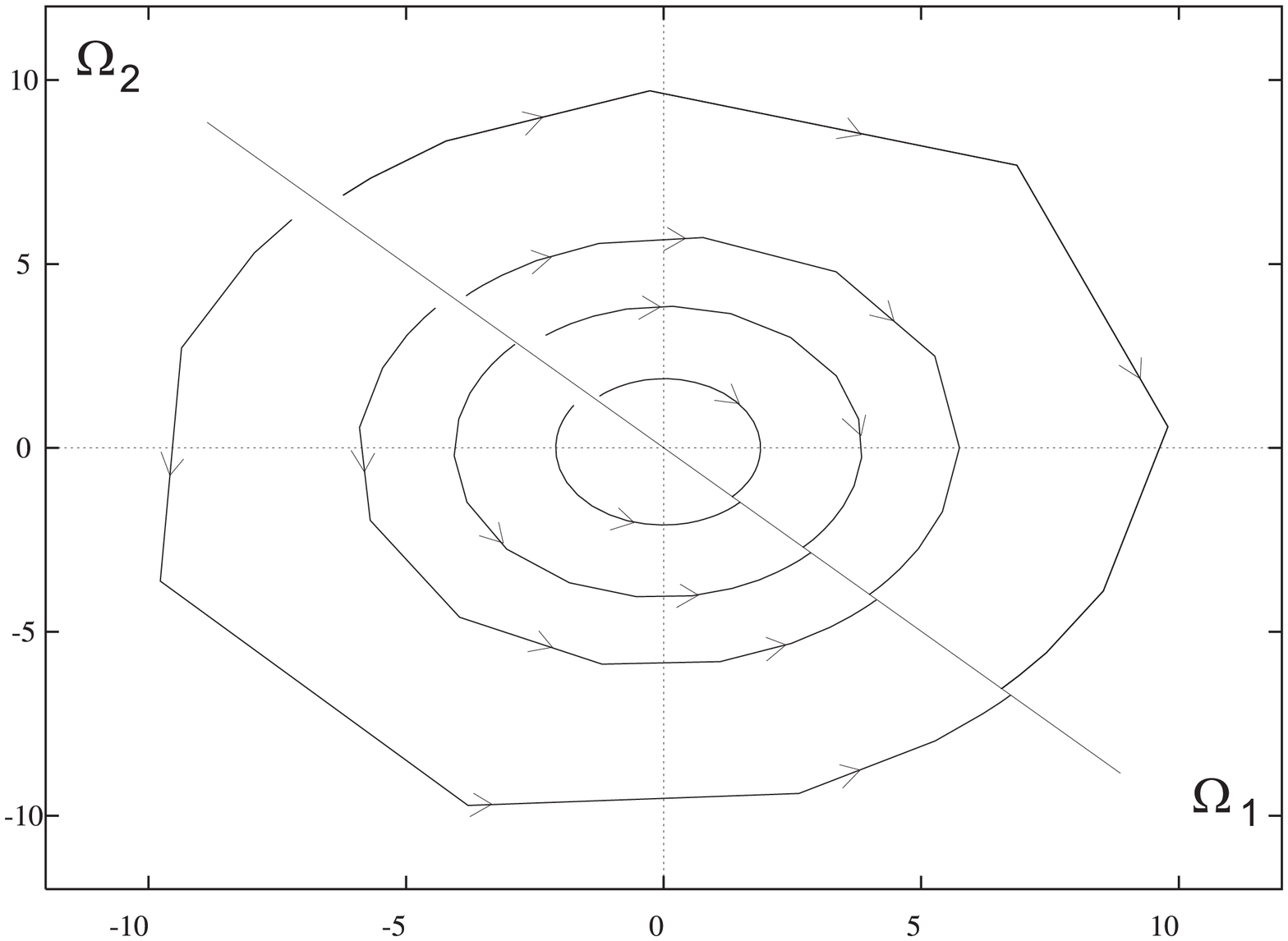, height=6.5cm}}

\begin{Remark} Equation \eqref{u} splits on two equations:
$$\widetilde{u}=-c\sqrt{u^2+1}+u\sqrt{c^2+1},\ \ \ \widetilde{u}=-c\sqrt{u^2+1}-u\sqrt{c^2+1}
$$
where $c=\frac{2\epsilon\sqrt{I_1I_2h}}{I_1I_2-h\epsilon^2}$.
The solutions given in Proposition \ref{t1} correspond to the first equation. One easily concludes that for the second
equation the limit when $\epsilon$ goes to zero is not defined well.
\end{Remark}

\begin{Remark} One can use the change of coordinates $x=I_{13}\Omega_1+I_{23}\Omega_2,
y=I_{23}I_{1}\Omega_1-I_{13}I_2\Omega_2$ also in the continuous case. Then equations \eqref{s1} become:
$$
\dot x=\frac{xy}{I_1 I_2},\ \ \dot y=-x^2.
$$
which are simpler then original ones.
One can easily see that Hirota-Kimura type discretizations of last equations are equations \eqref{xy}, as
linear change of variables commutes with Hirota-Kimura type discretization.
\end{Remark}
\section{Higher-dimensional case}

As it has already been mentioned, the higher-dimensional generalization of Suslov case was suggested by Kozlov and Fedorov (see\cite{FK, FZ}).
The configuration space is Lie group $SO(n)$ and the nonholonomic constraints are:
$$
\Omega_{ij}=0,\ \ \ 1\leq i,j\leq n-1.
$$
The equations of motion are:
\begin{equation}
\dot M=[M, \Omega] +\Lambda
\label{nsus}
\end{equation}
where $M=I\Omega+\Omega I$, and
$$
I=\left [\begin{matrix} I_{11}& 0& ...& I_{1n}\\
                        0&I_{22}&...&I_{2n}\\
                        \vdots&\vdots&...&\vdots\\
                        I_{1n}&I_{2n}&...&I_{nn}
        \end{matrix}
        \right],
        \quad
\Lambda=\left [\begin{matrix} 0& \lambda_{12}& ...&\lambda_{1,n-1}& 0\\
                        -\lambda_{12}&0&...&\lambda_{2,n-1}&0\\
                        \vdots&\vdots&...&\vdots\\
                        -\lambda_{1, n-1}&-\lambda_{2,n-1}&...&0&0\\
                        0&0&...&0&0
        \end{matrix}
        \right],
$$
From \eqref{nsus} closed systems of equations in $\Omega_{in},\ 1\leq i\leq n-1$ follows:
\begin{equation}
\begin{aligned}
(I_{ii}+I_{nn})&\dot\Omega_{in}=-I_{in}(\Omega_{1n}^2+...+\Omega_{n-1, n}^2)
&+(I_{1n}\Omega_{1n}+...+I_{n-1, n}\Omega_{n-1, n})\Omega_{in}.
\end{aligned}
\label{nsus1}
\end{equation}

Similarly as in three-dimensional case, we give the Hirota-Kimura discretization in $n$ dimensions by the following system of equations:
\begin{equation}
\begin{aligned}
(I_{ii}+I_{nn})(\widetilde{\Omega}_{in}-\Omega_{in})&=-\epsilon I_{in}(\widetilde{\Omega}_{1n}\Omega_{1n}+...+\widetilde{\Omega}_{n-1, n}\Omega_{n-1,n})\\
&+\epsilon(I_{1n}\widetilde{\Omega}_{1n}+...+I_{n-1, n}\widetilde{\Omega}_{n-1, n})\frac{\Omega_{in}}{2}\\
&+\epsilon(I_{1n}\Omega_{1n}+...+I_{n-1, n}\Omega_{n-1, n})\frac{\widetilde{\Omega}_{in}}{2}.
\end{aligned}
\label{dsn}
\end{equation}

As in three-dimensional case, map defined by \eqref{dsn} is explicit and unique-valued:
\begin{equation}
\left[\begin{matrix} \widetilde{\Omega}_{1n}\\ .\\.\\.\\\widetilde{\Omega}_{n-1,n}\end{matrix}\right]=A^{-1}
\left[\begin{matrix} \Omega_{1n}\\ .\\.\\.\\\Omega_{n-1,n}\end{matrix}\right]
\end{equation}
where
$$
A=\left[\begin{matrix} 1-\frac{\epsilon(I_{2n}\Omega_{2n}+...+I_{n-1,n}\Omega_{n-1,n})}{2(I_{11}+I_{nn})}&
...&\frac{\epsilon(2 I_{1n}\Omega_{n-1,n}-I_{n-1,n}\Omega_{1n})}{2(I_{11}+I_{nn})}\\
\frac{\epsilon(2 I_{2n}\Omega_{1n}-I_{1,n}\Omega_{2n})}{2(I_{22}+I_{nn})}&...&\frac{\epsilon (2 I_{2n}\Omega_{n-1,n}-I_{n-1,n}\Omega_{2n})}{2(I_{22}+I_{nn})}\\
\vdots\\
\frac{\epsilon(2I_{n-1,n}\Omega_{1n}-I_{1n}\Omega_{n-1,n})}{2(I_{n-1,n-1}+I_{nn})}&
...&1-\frac{\epsilon(I_{1n}\Omega_{1n}+...+I_{n-2,n}\Omega_{n-2,n})}{2(I_{n-1,n-1}+I_{nn})}
\end{matrix}
\right]
$$
Let us mention that equations \eqref{dsn} have a particular solution $\frac{\Omega_{1n}}{I_{1n}}=...=\frac{\Omega_{n-1,n}}{I_{n-1,n}}=c=const$, which
corresponds to motion with constant angular velocity.

\section*{Acknowledgments}

Our research is supported by the Project 144014 of the Ministry of Science of Republic of Serbia. We use the opportunity to thank Prof. Yu. N. Fedorov
for useful discussions on background of the Suslov problem and its discretization from \cite{FZ}.

\smallskip
\noindent
Vladimir Dragovi\'c,\\
Mathematical Institute SANU\\
Kneza Mihaila 36, Belgrade\\
Serbia,\\
Grupo de Fisica Matematica,\\
Complexo Interdisciplinar da Universidade de Lisboa\\
Av. Prof. Gama Pinto, 2\\
PT-1649-003 Lisboa\\
Portugal\\

\noindent Borislav Gaji\' c\\
Mathematical Institute SANU\\
Kneza Mihaila 36, Belgrade\\
Serbia\\

\begin{thebibliography}{99}

\bibitem{CM} J. Cort\' es, S. Mart\' inez: Nonholonomic integrators, {\it Nonlinearity}, {\bf 14}, (2001), 1365-1392.

\bibitem{DGJ} V. Dragovi\' c, B. Gaji\' c, B. Jovanovi\' c: Generalizations of
classical integrable nonholonomic rigid body motion, {\it J. Phys. A: Math. Gen.}, {\bf 31}, (1998), 9861-9869

\bibitem{FK} Yu. N. Fedorov, V. V. Kozlov: Various Aspects of $n$-Dimensional Rigid Body Dynamics,
{\it Amer. Math. Soc. Transl}, Ser. 2, {\bf 168}, (1995), 141-171

\bibitem{FZ} Yu. N. Fedorov, D. V. Zenkov: Discrete nonholonomic LL systems
on Lie groups, {\it Nonlineairy}, {\bf 18}, (2005), 2211-2241.

\bibitem{HK1} R. Hirota, K. Kimura: Discretization of Euler top, {\it Jour. Phys. Soc. Japan},
{\bf 69}, (2000), 627-630.

\bibitem{HKY} R. Hirota, K. Kimura, H. Yahagi: How to find conserved quantities of nonlinear discrete equations,
{\it J. Phys. A: Math. Gen.}, {\bf34}, (2001), 10377-10386.

\bibitem{J1} B. Jovanovi\' c: Some multidimensional integrable cases of nonholonomic rigid body dynamics,
{\it Reg. Chaotic Dynamics}, {\bf 8}, (2003),  125-132.

\bibitem{J} B. Jovanovi\' c: Geometry and integrability of Euler-Poincar\' e-Suslov equations, {\it Nonlinaerity},
{\bf 14}, (2001), 1555-1567

\bibitem{HK2}K. Kimura, R. Hirota: Discretization of the Lagrange top, {\it Jour. Phys. Soc. Japan},
{\bf 69}, (2000), 3193-3199.

\bibitem{K} V. V. Kozlov: On the integration theory of equations of nonholonomic mechanics, {\it Adv. Mech},
{\bf 8}, (1985), 85-107, [in Russian].

\bibitem {LMS} M. de Le\' on, D. Mart\' in de Diego, A. Santamar\' ia Merino: Geometric integrators and Nonholonomic Mechanics, {\it J. Math. Phys.},
45(3), (2004), 1042-1064.
q
\bibitem {MV} J. Moser, A.P. Veselov: Discrete version of some classical integrable systems
and factorization of matrix polynomials, {\it Commun. Math. Phys}, {\bf 139}, (1991), 217-243.

\bibitem {PS} M. Petrera, Yu. Suris: On the Hamiltonian structure of Hirota-Kimura
discretization of the Euler top, arXiv:math-ph:0707.4382v1, (2007).

\bibitem {Sur} Yu. Suris: {\it The problem of integrable discretization: Hamiltonian
approach}, Progres in Mathematics, {\bf 219}, Birkh\"{a}user Verlag, Basel, (2003).

\bibitem {Sus} G. Suslov: {\it Theoretical mechanics}, Gostekizdat, Moskow-Leningrad, (1946), [in Russian].

\bibitem {ZB} D. V. Zenkov, A. M. Bloch: Dynamics of the $n$-Dimensional Suslov problem, {\it J. Geom. Phys.}, {\bf 34}, (2000), 121-136.

\end{thebibliography}
\end{document}